\documentclass{amsart}

\newtheorem{theorem}{Theorem}[section]
\newtheorem{lemma}[theorem]{Lemma}

\theoremstyle{definition}
\newtheorem{problem}{Problem}
\newtheorem{corollary}{Corollary}
\newtheorem{definition}{Definition}
\newtheorem{proposition}{Proposition}
\newtheorem{example}{Example}

\theoremstyle{remark}
\newtheorem{remark}[theorem]{Remark}

\numberwithin{equation}{section}

\usepackage{amsmath,amsfonts,amssymb,color,amsthm}
\usepackage{epsfig}
\usepackage{algorithm}
\usepackage{algorithmic}
\usepackage{array}
\usepackage{multirow}
\usepackage{epstopdf}
\usepackage{tikz}
\usepackage{relsize}
\usetikzlibrary{shapes,arrows}
\usepackage{pmat}
\usepackage{epstopdf}
\usepackage{tikz}
\usepackage{flushend}
\usetikzlibrary{shapes}
\usetikzlibrary{arrows, decorations.markings}
\usepackage{scalerel}
\usepackage{lineno,hyperref}

\makeatletter
\g@addto@macro{\endabstract}{\@setabstract}
\newcommand{\authorfootnotes}{\renewcommand\thefootnote{\@fnsymbol\c@footnote}}%
\makeatother
\begin{document}
\begin{center}
  \LARGE 
  On the Computational Complexity of the Secure State-Reconstruction Problem\footnote{This work was funded in part by the Army Research Laboratory under Cooperative Agreement W911NF-17-2-0196, by the UC-NL grant LFR-18-548554, by the NSF award 1740047, and by the NSF CAREER award 1653648.} \par \bigskip

  \normalsize
  \authorfootnotes
  Yanwen Mao\textsuperscript{1}, Aritra Mitra\textsuperscript{2},
  Shreyas Sundaram\textsuperscript{2} and Paulo Tabuada\textsuperscript{1} \par \bigskip

  \textsuperscript{1}Department of Electrical and Computer Engineering, \\ University of California, Los Angeles \\ Email: yanwen.mao@g.ucla.edu, tabuada@ucla.edu \par
  \textsuperscript{2}Department of Electrical Engineering, Purdue University, West Lafayette \\ Email: mitra14@purdue.edu, sundara2@purdue.edu\par \bigskip

  \today
\end{center}

\begin{abstract}
In this paper, we discuss the computational complexity of reconstructing the state of a linear system from sensor measurements that have been corrupted by an adversary. The first result establishes that the problem is, in general, NP-hard. We then introduce the notion of eigenvalue observability and show that the state can be reconstructed in polynomial time when each eigenvalue is observable by at least $2s+1$ sensors and at most $s$ sensors are corrupted by an adversary. However, there is a gap between eigenvalue observability and the possibility of reconstructing the state despite  attacks - this gap has been characterized in the literature by the notion of sparse observability. To better understand this, we show that when the $\mathbf{A}$ matrix of the linear system has unitary geometric multiplicity, the gap disappears, i.e., eigenvalue observability coincides with sparse observability, and there exists a polynomial time algorithm to reconstruct the state provided the state can be reconstructed.
\end{abstract}


\section{INTRODUCTION}
This paper is concerned with the detection of attacks on Cyber-Physical Systems (CPSs). The distributed nature of these large-scale systems often leads to increased vulnerabilities. Of particular concern are adversaries that exploit the distributed nature of CPSs to gain access to sensors and launch attacks by modifying their measurements~\cite{Cardenas,cpssecurity,Giraldo}. The most notorious example is the Stuxnet malware~\cite{5772960}, which attacked numerous industrial control systems.

Over the last decade, a significant amount of research has focused on reconstructing the state in the presence of sensor attacks - we will refer to this as the Secure State-Reconstruction (SSR) problem throughout the paper. The first experimental demonstration of a stealthy attack on a control system was reported in~\cite{Amin} and it was followed by the first theoretical results developed for special classes of systems~\cite{Sandberg,5717544}. Stealthy attacks were then formalized in~\cite{SSmith,7011176}. An important step in the conceptual understanding of these attacks was given in~\cite{PasqualettiC,fabio,5605238}, where the existence of such attacks was characterized by the system theoretic notion of zero-dynamics. 

In addition to detecting and identifying attacks, it is important to mitigate their effect by continuing to control the plant. Hence, researchers have invested a significant effort in developing algorithms to reconstruct the state since the papers~\cite{6120187,6727407}. However, the SSR problem is intrinsically an NP-hard problem (as we show in this paper). Based on how the NP-hardness is tackled, we classify the existing work in two classes: 1) brute force search~\cite{7171098,Anyang}, and 2) computationally efficient relaxations. The methods reported in the first class are better suited for small systems as the computational complexity grows combinatorially with the number of sensors. Noteworthy examples of the second class include: convex relaxations~\cite{6727407,Yong}, distributed detection filters~\cite{fabio}, specialized observers under sparsity constraints~\cite{7308014}, satisfiability modulo theory techniques~\cite{YasserSMT}, and safety envelopes~\cite{Tiwari}.

The distributed version of the SSR problem has also attracted a substantial amount of interest given the distributed nature of CPSs. Several authors have studied the problem of estimating a static vector from a set of corrupted measurements, either over a distributed sensor network \cite{chen1,sufinite}, or over a connected-on-average network~\cite{chen2}. A control-theoretic approach to distributed function calculation was developed in~\cite{5605238}. Follow-up works have analyzed the resilient consensus problem, both for discrete \cite{6481629}, and continuous-time \cite{6580005} systems. The work in \cite{Tseng} also evaluates this method in various network topologies. The problem of guaranteeing resilience in the context of distributed state estimation, when the state of the system evolves over time (based on potentially unstable dynamics) has been recently explored in \cite{deghat}, \cite{mitraCDC16}, and \cite{mitraAuto}. In particular, the authors in \cite{mitraAuto} develop a fully-distributed algorithm that reconstructs the evolving state despite attacks on certain sensors in the network.

Despite the wealth of literature on the security of CPSs, to the best of the authors’ knowledge, a detailed characterization of the complexity of the SSR problem is still lacking. On the one hand, the papers~\cite{6727407,Yong,fabio,7308014,YasserSMT,Tiwari} suggest that the SSR problem is computationally hard since they propose efficient relaxations to the problem. On the other hand, the paper~\cite{mitraAuto} implicitly proposes a polynomial-time solution to the SSR poblem for certain cases. These observations naturally call for a better understanding of the complexity of the SSR problem, which is precisely the goal of this paper. 

As we shall soon see, two alternate notions of observability, namely ``sparse observability" introduced in \cite{6727407,7308014} (see also~\cite{5605238} for an equivalent notion in continuous time), and ``eigenvalue observability" \cite{chen}, \cite{mitraTAC18}, will play key roles in our characterization of the SSR problem complexity. 
Our contributions are the following:
\begin{enumerate}
\item We show that the SSR problem is NP-hard.
\item We provide a decomposition that identifies portions of the state that can be reconstructed in polynomial time and portions that are NP-hard to reconstruct.
\item We offer a polynomial-time solution for the SSR problem under an eigenvalue observability assumption.
\item We show that checking sparse observability is coNP-complete.
\item We show that the notions of sparse observability and eigenvalue observability are equivalent when the geometric multiplicity of each eigenvalue of the system matrix $\mathbf{A}$ is 1.

\end{enumerate}

These results can be understood as follows. Although the SSR problem is NP-hard, in general, there may be portions of the state that can be reconstructed in polynomial time. We perform a system decomposition to identify these different portions of the state. 
In particular, when all the eigenvalues of the system matrix $\mathbf{A}$ have unitary geometric multiplicity, the decomposition results in scalar SSR problems. This establishes the equivalence between sparse observability, a necessary and sufficient condition for the SSR problem to be solvable, and eigenvalue observability, a sufficient condition for the existence of a polynomial time algorithm. Interestingly, even if the unitary geometric multiplicity condition is not satisfied, we may still check eigenvalue observability and, if successful, solve the SSR problem in polynomial time. When the system does not satisfy the eigenvalue observability condition, we conjecture that the SSR problem is intractable since even checking sparse observability is coNP-complete. This paper improves upon the preliminary results in~\cite{Yanwen} by introducing a decomposition technique that is key to the aforementioned contributions 1 and 2.

The rest of the paper is organized as follows. In Section \ref{sec:prelim}, we define the notation used throughout the paper. In Section \ref{sec:prob_form}, we introduce the system model and give a formal definition of the SSR problem, sparse observability, and eigenvalue observability. We prove that the SSR problem is NP-hard in Section \ref{sec: SSRhard}. This is then followed by a result on breaking the overall SSR problem into several smaller independent SSR problems. 
As a special case, we show in Section \ref{sec:Theorem1} that under an eigenvalue observability assumption, the SSR problem can be solved in polynomial time. While checking eigenvalue observability can be done in polynomial time, in Section \ref{sec:theorem2} we show that checking sparse observability is coNP-complete. We connect these two notions in Section \ref{sec:connection} by showing that they are equivalent when the geometric multiplicity of each eigenvalue of the system matrix $\mathbf{A}$ is $1$. Finally, we conclude the paper in Section \ref{sec:conclusion}.

\section{PRELIMINARIES AND NOTATIONS}
\label{sec:prelim}

The cardinality of a finite set $\mathcal{I}=\{\mathbf{i}_1,\dots,\mathbf{i}_p\}$ is denoted by $|\mathcal{I}|=p$. For matrices $\mathbf{Q}_{i_1},\dots,\mathbf{Q}_{i_p}$ over the same field and with the same number of columns,  we define the matrix $\mathbf{Q}_{\mathcal{I}}=\begin{bmatrix}\mathbf{Q}_{i_1}^T\vert \mathbf{Q}_{i_2}^T\vert \dots\vert \mathbf{Q}_{i_p}^T\end{bmatrix}^T$ by stacking the individual matrices vertically.

We use $\mathbb{R}$ to denote the field of real numbers, $\mathbb{Q}$ to denote the field of rational numbers, and $\mathbb{C}$ to denote the field of complex numbers. For a matrix $\mathbf{A}\in\mathbb{R}^{n\times n}$, we use $\mathrm{ker}~\mathbf{A}$ to denote the kernel of $\mathbf{A}$, $\mathrm{Im}(\mathbf{A})$ to denote the image of $\mathbf{A}$ and $\mathbf{A}\vert_V$ to denote the restriction of the linear map defined by $\mathbf{A}$ to the subspace $V$. We also denote by $\mathbf{A}(V)$ the set $\{y\in\mathbb{R}^n \vert y=\mathbf{A}x,x\in V\}$.  

Let $V$ be a vector space. The collection of vector spaces $\{V^j\}_{j=1,\hdots,r}$, with $V^j\subseteq V$, is said to be an internal direct sum of $V$, denoted by $V=\bigoplus_{j=1,\hdots,r}V^j$, if any vector $v\in V$ can be uniquely written as $v=v_1+\hdots+v_r$ with $v_j\in V^j$. The direct sum comes equipped with canonical inclusions $\imath_j:V^j\to V$ taking $v_j\in V^j$ to $\imath_j(v_j)=v_j\in V$,  and canonical projections $\pi_j:V\to V^j$ taking $v\in V$ to $\pi_j(v)=v_j\in V^j$.

As an example, consider $V=\mathbb{R}^4$ and let $V^1=\mathrm{Im}(\mathbf{M}_1)$, $V^2=\mathrm{Im}(\mathbf{M}_2)$, and $V^3=\mathrm{Im}(\mathbf{M}_3)$ where $\mathbf{M}_1$, $\mathbf{M}_2$, and $\mathbf{M}_3$ are the following linear transformations:
\begin{equation}
\label{matrix:M}
\mathbf{M}_1=\begin{bmatrix}
\phantom{-}2 & 0\\
-1 & 1\\
\phantom{-}1 & 1\\
\phantom{-}0 & 0
\end{bmatrix},\quad 
\mathbf{M}_2=\begin{bmatrix}
\phantom{-}0\\
\phantom{-}1\\
-1\\
\phantom{-}0
\end{bmatrix}, \quad 
\mathbf{M}_3=\begin{bmatrix}
-1 \\
\phantom{-}1\\
\phantom{-}0\\
\phantom{-}1
\end{bmatrix}. 
\end{equation}
The collection $\{V^1,V^2,V^3\}$ is an internal direct sum of $V$ since all the column vectors are linearly independent. The canonical inclusions $\imath_j$ can be represented by $\mathbf{I}_4\vert_{V^j}$, the identity matrix $\mathbf{I}_4$ of order 4 restricted to the subspace $V^j$, since $\imath_j$ maps any vector $v\in V^j$ to $v\in V$.
Conversely, the canonical projections $\pi_j$ are represented by the matrices $\mathbf{P}_j=\mathbf{M}_i\mathbf{U}_j\mathbf{M}^{-1}$, where $\mathbf{U}_1=\begin{bmatrix}1 & 0 & 0 & 0 \\ 0 & 1 & 0 & 0\end{bmatrix}$, $\mathbf{U}_2=\begin{bmatrix}0 & 0 & 1 & 0\end{bmatrix}$, $\mathbf{U}_3=\begin{bmatrix}0 & 0 & 0 & 1\end{bmatrix}$, as well as $\mathbf{M}=\begin{bmatrix}\mathbf{M}_1 & \mathbf{M}_2 & \mathbf{M}_3\end{bmatrix}.$

Let $V=\bigoplus_{j=1,\hdots,r}V^j$, $W=\bigoplus_{j=1,\hdots,r}W^j$, and consider a linear map  $F:V\to W$  satisfying $F(V^j)\subseteq W^j$. Then, the linear map $F^{(j)}:V^j\to W^j$ defined by $F^{(j)}=\pi_j\circ F\circ \imath_j$ satisfies: 
\begin{eqnarray}
\label{ComuteP}
F^{(j)}\circ \pi_j&=&\pi_j \circ F\\
\label{ComuteI}
\imath_j\circ F^{(j)}&=&F\circ\imath_j,
\end{eqnarray}
where $\circ$ denotes function composition.

Continuing with our example, let $\mathbf{F}$ be represented by the matrix: 
\begin{equation}
\label{matrix:F}
\mathbf{F}=\frac{1}{2}\begin{bmatrix}
\phantom{-}2 & \phantom{-}0 &\phantom{-}0 &-4\\
\phantom{-}1 & \phantom{-}3 & -1 & \phantom{-}4\\
-1 & -1 & \phantom{-}3 & \phantom{-}0\\
\phantom{-}0 & \phantom{-}0 &\phantom{-}0 &\phantom{-}6
\end{bmatrix},
\end{equation}
and note that $\mathbf{F}(V^j)\subseteq V^j$. The maps $\mathbf{F}^{(j)}$ are then given by $\mathbf{F}^{(1)}=\mathbf{P}_1\mathbf{F}\circ\imath_1=\mathbf{P}_1\mathbf{F}\vert_{V^1}=\mathbf{I}_4\vert_{V^1}, \mathbf{F}^{(2)}=\mathbf{P}_2\mathbf{F}\circ\imath_2=\mathbf{P}_2\mathbf{F}\vert_{V^2}=2\mathbf{I}_4\vert_{V^2}$, as well as $\mathbf{F}^{(3)}=\mathbf{P}_3\mathbf{F}\circ\imath_3=\mathbf{P}_3\mathbf{F}\vert_{V^3}=3\mathbf{I}_4\vert_{V^3}$. Since the vector subspaces $V^j$ are the generalized eigenspaces of $\mathbf{F}$ corresponding to each different eigenvalue, the matrices $\mathbf{F}^{(j)}$ are simply the identity matrix restricted to $V^j$ multiplied by the corresponding eigenvalue. 

We denote by $\lambda_1,\hdots,\lambda_r\in \mathbb{C}$ the (counted without repetition) eigenvalues of $\mathbf{A}$ and by $sp(\mathbf{A})=\{\lambda_1,\hdots,\lambda_r\}$ its spectrum. The algebraic multiplicity of an eigenvalue $\lambda_j$, denoted by $\alpha(\lambda_j)$, is the number of times (counted with repetition) that $\lambda_j$ is a solution of $\det (\mathbf{A}-\lambda_j\mathbf{I}_n)=0$. The geometric multiplicity of an eigenvalue $\lambda_j$, denoted by $\gamma(\lambda_j)$, is the dimension of the vector space $\ker (\mathbf{A}-\lambda_j \mathbf{I}_n)$. We denote the space of generalized eigenvectors associated with $\lambda_j$, $\ker (\mathbf{A}-\lambda_j \mathbf{I}_n)^{\alpha(\lambda_j)}$,  by $V_j$. Note that $V_j$ has dimension $\alpha(\lambda_j)$ and  $\gamma(\lambda_j)$ Jordan chains. 

Given a vector $\mathbf{b}\in \mathbb{R}^n$, we denote by $\Vert \mathbf{b}\Vert_0$ the number of non-zero entries in  $\mathbf{b}$.

\section{PROBLEM FORMULATION}
\label{sec:prob_form}
\subsection{System Model}
Consider a discrete-time linear time-invariant system under sensor attacks of the following form: 
\begin{eqnarray}
\mathbf{x}(k+1)&=&\mathbf{Ax}(k) \label{eqn:modelsys}\\
\mathbf{y}_{i}(k)&=&\mathbf{C}_i\mathbf{x}(k)+\mathbf{e}_i(k), \label{eqn:modelmeasure}
\end{eqnarray}
where $\mathbf{x}(k)\in\mathbb{R}^{n}$ and $\mathbf{y}_{i}(k) \in {\mathbb{R}}^{p_i}$ represent the state of the system and the measurement acquired by sensor $i$ respectively. The vector $\mathbf{e}_i(k)\in\mathbb{R}^{p_i}$ models the attack on sensor $i$. If sensor $i$ is attacked by an adversary, then $\mathbf{e}_i(k)$ can be arbitrary, otherwise, $\mathbf{e}_i(k)$ remains zero for any $k$. Let $\mathcal{V}$ denote the set of sensors, and let $N=|\mathcal{V}|$. We use $\mathbf{C}={\begin{bmatrix}\mathbf{C}^T_{1} \vert \mathbf{C}^T_{2} \vert \cdots \vert \mathbf{C}^T_{N}\end{bmatrix}}^T$ to denote the collection of the sensor observation matrices, $\mathbf{y}(k)={\begin{bmatrix}\mathbf{y}^T_{1}(k) & \cdots & \mathbf{y}^T_{N}(k)\end{bmatrix}}^T$ and $\mathbf{e}(k)={\begin{bmatrix}\mathbf{e}^T_{1}(k) & \cdots & \mathbf{e}^T_{N}(k)\end{bmatrix}}^T$ to represent the collective measurement vector and the collective attack vector, respectively.

We define $\mathcal{O}_i=\begin{bmatrix}\mathbf{C}_i^T \vert (\mathbf{C}_i\mathbf{A})^T \vert \dots \vert (\mathbf{C}_i\mathbf{A}^{\tau_i-1})^T \end{bmatrix}^T$ to be the observability matrix of sensor $i$ with $\tau_i$ being the observability index of the pair $(\mathbf{A},\mathbf{C}_i)$. We also define two more vectors $\mathbf{Y}_i=\begin{bmatrix}\mathbf{y}_i^T(0) & \dots & \mathbf{y}_i^T(\tau_i-1) \end{bmatrix}^T$ and  $\mathbf{E}_i=\begin{bmatrix}\mathbf{e}_i^T(0) & \dots & \mathbf{e}_i^T(\tau_i-1) \end{bmatrix}^T$ to be the collection of measurements and attacks of sensor $i$ over the time horizon $[0, \tau_i-1]$, respectively. An equivalent expression for the measurements is:
\begin{equation}
\label{modelMeasure}
    \mathbf{Y}_i=\mathcal{O}_i\mathbf{x}(0)+\mathbf{E}_i.
\end{equation}

In the remainder of the paper, we drop the time indices to simplify notation.

\subsection{The Secure State-Reconstruction Problem}

\begin{problem} (Secure state-reconstruction)\\ 
{\bf Input:} Matrices $\mathbf{A}\in\mathbb{R}^{n\times n}$, $\mathbf{C}_i \in \mathbb{R}^{p_i\times n},~i=1,\dots,N,$ and a set of vectors ${\mathbf{Y}_i\in \mathbb{R}^{p_i\tau_i}},~i=1,\dots,N.$\\
{\bf Question:} Find a vector $\mathbf{x}\in \mathbb{R}^n$ and a set $\mathcal{I}$ of minimal cardinality such that $\mathbf{Y}_j=\mathcal{O}_j\mathbf{x}$ for all $j\notin \mathcal{I}$.
\end{problem}
In other words, the SSR problem requires the reconstruction of a state $\mathbf{x}$ and the simplest attack explanation in the form of the least number of attacked sensors. Note that when the solution $\mathbf{x}$ is unique, we have found the state of the linear system. Although uniqueness of solutions is essential when handling attacks, we can study the complexity of the SSR problem independently of the number of solutions. To make this clear, we will explicitly state the uniqueness requirements when needed.

\subsection{Sparse Observability and Eigenvalue Observability}
\label{sec:sparse_obs}
The notions of sparse observability and eigenvalue observability are instrumental to the results in this paper.  

\begin{definition}[Sparse observability index]
The sparse observability index of the pair $(\mathbf{A},\mathbf{C})$ in system~\eqref{eqn:modelsys}-\eqref{eqn:modelmeasure} is the largest integer $k$ such that $\mathrm{ker}~\mathcal{O}_{\mathcal{V}\backslash \mathcal{K}}=\{0\}$ for any $\mathcal{K}\subseteq \mathcal{V},~|\mathcal{K}|\leq k$. When the sparse observability index is $r$, we say that system~\eqref{eqn:modelsys}-\eqref{eqn:modelmeasure} is $r-$sparse observable.
\label{def:sparse_obs_index}
\end{definition}

It is proved in \cite{6727407,7308014} (see also~\cite{7171098} for a similar notion in continuous time) that the possibility of uniquely reconstructing the state $\mathbf{x}(k)$ is characterized by the sparse observability index.

\begin{theorem}[\cite{6727407,7171098,7308014}]
Consider the linear system \eqref{eqn:modelsys}-\eqref{eqn:modelmeasure} where at most $s$ sensors are subject to attacks. The state $\mathbf{x}(k)$ can be uniquely reconstructed if and only if the sparse observability index of the pair $(\mathbf{A},\mathbf{C})$ is at least $2s$.
\end{theorem}

In view of this result, computing the sparse observability index of a system is of great interest since it characterizes the maximum number of arbitrary sensor attacks that can be tolerated without compromising the ability to uniquely reconstruct the state. 

In addition to sparse observability, we will require the notion of eigenvalue observability \cite{chen,mitraTAC18}.

\begin{definition}[Eigenvalue observability index]
We say that an eigenvalue $\lambda \in sp(\mathbf{A})$ is observable w.r.t. sensor $i$ if the linear map defined by
$\begin{bmatrix}\mathbf{A}-\lambda\mathbf{I}_n\\ \mathbf{C}_i\end{bmatrix}$ is injective. 

If the above condition is satisfied, we say that ``sensor $i$ can observe the states in the generalized eigenspace corresponding to $\lambda$", or briefly, we say ``sensor $i$ can observe eigenvalue $\lambda$". Let the set of all sensors that can observe an eigenvalue $\lambda$ be denoted $\mathcal{S}_{\lambda}$. The eigenvalue observability index of system~\eqref{eqn:modelsys}-\eqref{eqn:modelmeasure} is the largest integer $k$ such that each eigenvalue of the matrix $\mathbf{A}$ is observable by at least $k + 1$ distinct sensors. When the eigenvalue observability index is $k$, we say that system~\eqref{eqn:modelsys}-\eqref{eqn:modelmeasure} is $k$-eigenvalue observable. 
\end{definition}

We study the SSR problem under the following assumptions.

\textbf{Assumption 1:} For each sensor $i \in \{1, \ldots, N\}$ under attack, the adversary can only manipulate sensor $i$'s measurements through the signal $\mathbf{e}_i(k)$ in \eqref{eqn:modelmeasure}.

\textbf{Assumption 2:} The adversary is omniscient, i.e., we assume the adversary has full knowledge of the system state, measurements, and plant model. Moreover, all the attacked sensors are allowed to work cooperatively.

\section{SSR IS HARD}
\label{sec: SSRhard}
Fawzi \emph{et al.} established in~\cite{6727407} a connection between the SSR problem and compressed sensing by drawing inspiration from the ideas of Candes and Tao in~\cite{1542412}. We take this approach further by also using the ideas in~\cite{1542412} to establish that the SSR problem is NP-hard. To do so, we first define the compressed sensing problem. 

\begin{problem} (Compressed sensing)\\ 
{\bf Input:} A full row rank matrix $\mathbf{F}\in \mathbb{Q}^{m\times n}$, a vector $\mathbf{b}\in \mathbb{Q}^m$.\\
{\bf Question:} Find the sparsest solution of $\mathbf{Fx}=\mathbf{b}$.
\end{problem}

The compressed sensing problem yields the solution to the minimization problem:
\begin{eqnarray}
\min_{\mathbf{x}} & \left\lVert \mathbf{x} \right\rVert_0 \\ \notag
\textrm{s.t.} & \mathbf{Fx}=\mathbf{b}. \\ \notag
\end{eqnarray}

\begin{theorem}[\cite{6727407}]
The SSR problem is NP-hard.
\end{theorem}
\begin{proof}
Given an instance of the compressed sensing problem, we generate an instance of the SSR problem as follows. Let the system matrix be of the form $\mathbf{A}=\mathbf{I}_n$, and the collective observation matrix $\mathbf{C}$ satisfy $\mathrm{Im}\mathbf{C}=\mathrm{ker}~\mathbf{F}$. Let the measurements of the sensors be scalar-valued, i.e., let $\mathbf{C}_i$ be the $i$-th row of $\mathbf{C}$. Note that based on the above $\mathbf{A}$ matrix, the observability index for each sensor $i \in \{1, \ldots, N\}$ is given by $\tau_i = 1$, and thus $\mathcal{O}_i = \mathbf{C}_i$.  Finally, let $\mathbf{Y}$ be any solution to the equation $\mathbf{F}\mathbf{Y}=\mathbf{b}$. Since the linear equation $\mathbf{F}{\mathbf{Y}}=\mathbf{b}$ is underdetermined, finding a solution ${\mathbf{Y}}$ can be done in polynomial time~\cite{Alan}. For each $i \in \{1, \ldots, N\}$, set $\mathbf{Y}_i$ to be the $i$-th row of $\mathbf{Y}$. Thus, given an instance of the compressed sensing problem, the instance of the SSR problem described above can be constructed in polynomial time. 

 The SSR problem for the constructed instance degenerates to:
\begin{eqnarray}
\label{Opti}
\min_{\mathbf{x,e}} & \left\lVert \mathbf{e} \right\rVert_0 \\ \notag
\textrm{s.t.} & \mathbf{Cx}+\mathbf{e}=\mathbf{Y}. \\ \notag
\end{eqnarray}

We now show these two problems have the same solution.  It is simple to see that  any solution $(\mathbf{x,e})$ of $\mathbf{Cx}+\mathbf{e}=\mathbf{Y}$ provides a solution to $\mathbf{Fe}=\mathbf{b}$, since by applying $\mathbf{F}$ we obtain:
\begin{eqnarray} 
\mathbf{F}(\mathbf{Cx}+\mathbf{e})&=&\mathbf{F}{\mathbf{Y}} \\ \notag
\Leftrightarrow \mathbf{Fe}&=&\mathbf{b}.
\end{eqnarray}
To prove the converse, we show that for every $\mathbf{e}$ such that $\mathbf{Fe}=\mathbf{b}$, there exists some $\mathbf{x}$ satisfying $\mathbf{Cx}+\mathbf{e}=\mathbf{Y}$. Recalling that $\mathbf{F}{\mathbf{Y}}=\mathbf{b}$, we obtain $\mathbf{F}(\mathbf{Y}-\mathbf{e})=\mathbf{0}$, i.e., ${\mathbf{Y}}-\mathbf{e}\in \mathrm{ker}~\mathbf{F}$. Since $\mathrm{ker}~\mathbf{F}=\mathrm{Im} \mathbf{C}$, there exists an $\mathbf{x}$ such that $\mathbf{Cx}=\mathbf{Y}-\mathbf{e}$, as desired.

Noticing that the equations $\mathbf{Fe}=\mathbf{b}$ and $\mathbf{Cx}+\mathbf{e}=\mathbf{Y}$ have the same solutions for $\mathbf{e}$, we conclude that  they also have the same sparsest solution. In other words, if there exists an algorithm $\mathcal{A}$ that solves the SSR problem for the specific instance constructed by us, such an algorithm will also yield a solution to the given instance of the compressed sensing problem. It then follows that since the compressed sensing problem is NP-hard~\cite{natarajan}, the secure state reconstruction problem is also NP-hard.
\end{proof}

\section{SYSTEM DECOMPOSITION}
\label{sec:decompose}
In the previous section, we proved that the SSR problem is in general NP-hard. This means there does not exist a polynomial-time solution unless $P=NP$. Despite this fact, we show in this section how to decompose the SSR problem into smaller instances. In the next section, we identify which of these smaller instances are NP-hard, and which ones are solvable in polynomial time.

\begin{lemma}
	\label{TechnicalLemma}
	Assume the existence of a collection of vector spaces $\{X^j\}_{j=1,\hdots,r}$ satisfying:
	\begin{enumerate}
		\item $\mathbb{C}^n=\bigoplus_{j=1,\hdots,r}X^j$;
		\item $\mathbf{A}(X^j)\subseteq X^j$ for $j=1,\hdots,r$;
		\item $\mathcal{O}_i(\mathbb{C}^n)=\bigoplus_{j=1,\hdots,r}\mathcal{O}_i^j(X^j)\quad \text{for~} i=1,\hdots,p,$
	\end{enumerate}
	then for any $\mathbf{Y}_i$, a solution $\mathbf{x}$ of the equation:
	\begin{equation}
	\label{SingleEq}
	\mathbf{Y}_i=\mathcal{O}_i\mathbf{x},
	\end{equation}
	whenever it exists, can be written as $\mathbf{x}=\mathbf{x}_1+\mathbf{x}_2+\hdots+\mathbf{x}_r$ with $\mathbf{x}_j=\pi_j(\mathbf{x})\in X^j$ given by the solution of:
	\begin{equation}
	\label{MultipleEq}
	\mathbf{Y}_i^j=\mathcal{O}_i^j\mathbf{x}_j,
	\end{equation}
	for $\mathbf{Y}_i^j=\pi_j(\mathbf{Y}_i)\in \mathcal{O}_i^j(X^j)$ and $\mathcal{O}_i^j=\pi_j\circ\mathcal{O}_i\circ \imath_j$.
\end{lemma}

\begin{proof}
	Let $\mathbf{x}_j$ be the solution of~\eqref{MultipleEq} and note that:
	\begin{equation}
	\mathbf{Y}_i^j=\mathcal{O}_i^j \mathbf{x}_j \Rightarrow \imath_j(\mathbf{Y}_i^j)=\imath_j\circ \mathcal{O}_i^j (\mathbf{x}_j)=\mathcal{O}_i\circ \imath_j(\mathbf{x}_j)=\mathcal{O}_i\mathbf{x}_j,  
	\end{equation}
	where the third equality follows from~\eqref{ComuteI}. By summing over $j$ we obtain:
	\begin{equation}
	\mathbf{Y}_i=\sum_{j=1}^{r}\imath_j(\mathbf{Y}_i^j)=\sum_{j=1}^{r}\mathcal{O}_i\mathbf{x}_j=\mathcal{O}_i\sum_{j=1}^{r}\mathbf{x}_j=\mathcal{O}_i \mathbf{x}.
	\end{equation}
	Hence, the solutions to~\eqref{MultipleEq} provide a solution to~\eqref{SingleEq}. Consider now~\eqref{SingleEq}:
	\begin{equation}
	\begin{aligned}
	\mathbf{Y}_i= & \mathcal{O}_i\mathbf{x}\Rightarrow  \pi_j(\mathbf{Y}_i)=\pi_j\circ \mathcal{O}_i(\mathbf{x}) \\
	\Rightarrow & \mathbf{Y}_i^j=\mathcal{O}_i^j\circ \pi_j(\mathbf{x})=\mathcal{O}_i^j \mathbf{x}_j.
	\end{aligned}
	\end{equation}
	where the third equality follows from~\eqref{ComuteP}. Hence, if $\mathbf{x}$ is a solution to~\eqref{SingleEq}, then $\mathbf{x}_i$ is a solution to~\eqref{MultipleEq}.
\end{proof}
Intuitively, we treat the state-space $\mathbb{R}^n$ as the direct sum of multiple subspaces. If the images of these subspaces under the linear map $\mathcal{O}_i$ are pairwise non-overlapping, we are able to project the state vector $\mathbf{x}$ onto these subspaces, project the measurement $\mathbf{Y}_i$ onto the image under the linear map $\mathcal{O}_i$ of these subspaces, and then establish a one-to-one correspondence between the projected state vector and the projected measurement. This effectively decomposes the original problem into $r$ sub-problems, each of dimension $\mathrm{dim}(X^j)$. As formalized in the next result, the spaces $X^j$ can always be taken to be the generalized eigenspaces of $\mathbf{A}$.

\begin{proposition}
	\label{TechnicalProposition}
	The generalized eigenspaces $V^1,V^2,\dots,V^r$ of $\mathbf{A}$ satisfy properties (1)-(3) in Lemma~\ref{TechnicalLemma}.
\end{proposition}

\begin{proof}
	Properties (1) and (2) in Lemma~\eqref{TechnicalLemma} follow directly from the definition of generalized eigenspace. To simplify notation,  we will drop the sensor index $i$ in this proof.
	
	It also follows from the definition of generalized eigenspace that $\cup_{j=1,\hdots,r}V^j$ spans $\mathbb{C}^n$. Therefore, the set $\cup_{j=1,\hdots,r}\mathcal{O}(V^j)$ spans $\mathcal{O}(\mathbb{C}^n)$. Given this, to conclude property (3) we only need to show:
	$$\mathcal{O}(V^j)\cap \mathcal{O}(V^k)=\{0\},\qquad \forall j\ne k.$$
	Moreover, it suffices to show that for any $\mathbf{x}_j\in V^j$ and $\mathbf{x}_k\in V^k$, with $j\ne k$, the equality $\mathcal{O}(\mathbf{x}_j+\mathbf{x}_k)=0$ can only be satisfied if $\mathcal{O}\mathbf{x}_j=0=\mathcal{O}\mathbf{x}_k$.
	
	
	We have the following sequence of equalities that is explained thereafter:
	\begin{eqnarray}
	0 & = & \mathcal{O}(\mathbf{x}_j+\mathbf{x}_k)\\
	& = & \mathcal{O}(\mathbf{A}-\lambda_k \mathbf{I}_n)^{\alpha(\lambda_k)}(\mathbf{x}_j+\mathbf{x}_k)\\
	& = & \mathcal{O}(\mathbf{A}-\lambda_k \mathbf{I}_n)^{\alpha(\lambda_k)}(\mathbf{x}_j)\\
	& = & \mathcal{O}\mathbf{x}_j.
	\end{eqnarray}
	The second step follows from $\ker\mathcal{O}\subseteq \ker \mathcal{O}(\mathbf{A}-\lambda_k \mathbf{I}_n)^{\alpha(\lambda_k)}$, the third step follows from \mbox{$\mathbf{x}_k\in V^k=\ker (\mathbf{A}-\lambda_k \mathbf{I}_n)^{\alpha(\lambda_k)}$}, and the fourth from the following sequence of steps:
	\begin{eqnarray}
	\dim\ker\mathcal{O}\big\vert_{V^j} &\leq& \dim\ker\mathcal{O}(\mathbf{A}-\lambda_k \mathbf{I}_n)^{\alpha(\lambda_k)}\big\vert_{V^j}\\
	&=& \dim\ker(\mathbf{A}-\lambda_k\mathbf{I}_n)^{\alpha(\lambda_k)}\big\vert_{V^j} \\ 
	&+& \dim\ker\mathcal{O}\big\vert_{(\mathbf{A}-\lambda_k \mathbf{I}_n)^{\alpha(\lambda_k)}{V^j}}\\
	&=& \dim\ker\mathcal{O}\big\vert_{(\mathbf{A}-\lambda_k \mathbf{I}_n)^{\alpha(\lambda_k)}{V^j}}\\
	&\leq& \dim\ker\mathcal{O}\big\vert_{V^j}.
	\end{eqnarray}
	The first step comes from $\ker\mathcal{O}\subseteq \ker\mathcal{O}(\mathbf{A}-\lambda_k\mathbf{I}_n)$. To show that the second step is true, we observe that \\ \mbox{$\mathrm{dim~}\mathrm{ker}~\mathbf{MN} = \mathrm{dim~}\mathrm{ker}~\mathbf{N} + \mathrm{dim~}\mathrm{ker}(\mathbf{M}\big\vert_{\mathbf{N}(\mathbb{C}^n)})$} for any matrices $\mathbf{M},\mathbf{N}\in \mathbb{C}^{n\times n}$. The third step comes from the map $(\mathbf{A}-\lambda_j \mathbf{I}_n)^{\alpha(\lambda_j)}\big\vert_{V^j}$ being injective if $j\ne k$, as the generalized eigenspaces $V^j$ and $V^k$ intersect only at the origin, and $\ker (\mathbf{A}-\lambda_j \mathbf{I}_n)^{\alpha(\lambda_j)}=V^j$. The fourth step follows by the $\mathbf{A}-$invariant nature of eigenspace $V^j$. This shows $\dim\ker\mathcal{O}\big\vert_{V^j}=\dim\ker\mathcal{O}(\mathbf{A}-\lambda_k \mathbf{I}_n)^{\alpha(\lambda_k)}\big\vert_{V^j}$ which, combined with $\ker\mathcal{O}\big\vert_{V^j} \subseteq \ker\mathcal{O}(\mathbf{A}-\lambda_k \mathbf{I}_n)^{\alpha(\lambda_k)}\big\vert_{V^j}$, can only hold when $\ker\mathcal{O}\big\vert_{V^j} = \ker\mathcal{O}(\mathbf{A}-\lambda_k \mathbf{I}_n)^{\alpha(\lambda_k)}\big\vert_{V^j}$. A symmetric argument can be used to show that $\mathcal{O}\mathbf{x}_k=0$ and the claim is thus proved.
\end{proof}
Combining Lemma \eqref{TechnicalLemma} and Proposition \eqref{TechnicalProposition} results in a decomposition of the sensor measurements in~\eqref{modelMeasure}:
\begin{eqnarray}
\mathbf{Y}_i^j=\mathcal{O}_i^j\mathbf{x}_j,\quad j=1,2,\dots,r,
\end{eqnarray}
where $\mathbf{Y}_i^j=\pi_j(Y_i)$ is the projection of measurement $\mathbf{Y}_i$ onto the vector space  $\mathcal{O}_i(V^j)$, the linear transformation $\mathcal{O}_i^j$ is defined by $\mathcal{O}_i^j=\pi_j\circ\mathcal{O}_i\circ \imath_j$, $\mathbf{x}_j$ is given by $\mathbf{x}_j=\pi_j(\mathbf{x})$, $\pi_j: \mathbb{R}^n\rightarrow V^j$ is the canonical projection and $\imath_j: V^j\rightarrow \mathbb{R}^n$ is the canonical inclusion.


\begin{theorem}
	\label{thm:decomformal}
	 A solution $\mathbf{x}$ of the SSR problem with inputs $\mathbf{A},\mathbf{C}_i,\mathbf{Y}_i$ is given by $\mathbf{x}=\mathbf{x}_1+\mathbf{x}_2+\dots +\mathbf{x}_j$ where $\mathbf{x}_i$ is the solution to the SSR problem with inputs $\mathbf{A}^{(j)}=\pi_j\circ \mathbf{A}\circ \imath_j$, $\mathbf{C}_i^j=\mathbf{C}_i\circ\imath_j$, $\mathbf{Y}_i^j$. 
\end{theorem}
\begin{proof}
	Follows directly from Lemma \ref{TechnicalLemma}, Proposition \ref{TechnicalProposition}, and the properties of generalized eigenspaces.
\end{proof}

Theorem \eqref{thm:decomformal} lays the theoretical foundation for decomposing the SSR problem with $n$ states into $r$ sub-problems of the form:
\begin{equation}
\begin{aligned}
\label{eqn:decompose}
\mathbf{x}_j(k+1) &= \mathbf{A}^{(j)}\mathbf{x}_j(k),\\
\mathbf{Y}_i^j(k) &= \mathcal{O}_i^j\mathbf{x}_j(k)+\mathbf{E}_i^j(k),
\end{aligned}
\end{equation}
each with $\alpha(\lambda_1),\alpha(\lambda_2),\dots,\alpha(\lambda_r)$ states. The attack vector $\mathbf{E}_i^j$ is identically zero when sensor $i$ is not under attack. The state of the original problem can be reconstructed by summing up the state reconstructions of each sub-problem.

We now illustrate the decomposition of \eqref{eqn:modelsys}-\eqref{eqn:modelmeasure} into \eqref{eqn:decompose} through an example. The matrix $\mathbf{A}$ is the same as the matrix $\mathbf{F}$ defined in \eqref{matrix:F} and the matrices $\mathbf{C}_i$ are given by:
\begin{eqnarray} \notag
\mathbf{C}_1=\begin{bmatrix}3&2&0&2\end{bmatrix}, &&\mathbf{C}_2=\begin{bmatrix}2&3&1&-1\end{bmatrix}, \\ \notag \mathbf{C}_3=\begin{bmatrix}2&2&0&0\end{bmatrix}, &&\mathbf{C}_4=\begin{bmatrix}2&3&-1&0\end{bmatrix}.
\end{eqnarray}
As we discussed below~\eqref{matrix:F}, the generalized eigenspaces of $\mathbf{A}$ are $V^1=\mathrm{Im}(\mathbf{M}_1),V^2=\mathrm{Im}(\mathbf{M}_2),$ and $V^3=\mathrm{Im}(\mathbf{M}_3)$ corresponding to eigenvalues $1,2,$ and $3$ respectively, where $\mathbf{M}_j$ are defined in~\eqref{matrix:M} for $j=1,2,3$. Also, recall that the projections $\pi_1,\pi_2,$ and $\pi_3$ are $\mathbf{P}_j=\mathbf{M}_j(\mathbf{M}_j^T\mathbf{M}_j)^{-1}\mathbf{M}_j^T$ for $j=1,2,3$. By definition, we have $\mathbf{x}_1=\mathbf{P}_1\mathbf{x}$, $\mathbf{x}_2=\mathbf{P}_2\mathbf{x}$, $\mathbf{x}_3=\mathbf{P}_3\mathbf{x}$, and $\mathbf{A}^{(1)}=\mathbf{P}_1\mathbf{A}\vert_{V^1}$, $\mathbf{A}^{(2)}=\mathbf{P}_2\mathbf{A}\vert_{V^2}$, $\mathbf{A}^{(3)}=\mathbf{P}_3\mathbf{A}\vert_{V^3}$. Hence the decomposition of $\mathbf{x}(k+1)=\mathbf{A}\mathbf{x}(k)$ is given by: $$\mathbf{P}_j\mathbf{x}(k+1)=(\mathbf{P}_j\mathbf{A}\vert_{V^j})(\mathbf{P}_j\mathbf{x}(k)),\quad j=1,2,3.$$

We now illustrate how to decompose the measurement equation $\mathbf{Y}_1(k)=\mathcal{O}_1\mathbf{x}(k)+\mathbf{E}_1(k)$ for sensor 1. The observability matrix $\mathcal{O}_1$ of sensor $1$ is given by: $$\mathcal{O}_1=\begin{bmatrix}3 & 2 & \phantom{-}0 & 2\\4 & 3 & -1 & 4\\ 6 & 5 & -3 & 10\\ 10 & 9  & -7 & 28\end{bmatrix}.$$
We first compute the projections $\widetilde\pi_1^1,\widetilde\pi_1^2$ and $\widetilde\pi_1^3$ that map $\mathcal{O}_1(\mathbb{R}^4)$ to $\mathcal{O}_1(V^1),\mathcal{O}_1(V^2),$ and $\mathcal{O}_1(V^3)$, respectively. To do this, we define the matrices: $$\widetilde{\mathbf{M}}_1=\begin{bmatrix}1\\1\\1\\1\end{bmatrix},~\widetilde{\mathbf{M}}_2=\begin{bmatrix}1\\2\\4\\8\end{bmatrix},~\text{and}~\widetilde{\mathbf{M}}_3=\begin{bmatrix}1\\3\\9\\27\end{bmatrix},$$ which satisfy $\mathcal{O}_1(V^1)=\mathrm{Im}(\widetilde{\mathbf{M}}_1)$, $\mathcal{O}_1(V^2)=\mathrm{Im}(\widetilde{\mathbf{M}}_2)$, and $\mathcal{O}_1(V^3)=\mathrm{Im}(\widetilde{\mathbf{M}}_3)$. We also remark that the collection $\{\mathcal{O}_1(V^1),\mathcal{O}_1(V^2),\mathcal{O}_1(V^3)\}$ is an internal direct sum of the vector space $\mathcal{O}_1(\mathbb{R}^4)$. Therefore, by defining $\widetilde{\mathbf{M}}=\begin{bmatrix}\widetilde{\mathbf{M}}_1 & \widetilde{\mathbf{M}}_2 & \widetilde{\mathbf{M}}_3\end{bmatrix}$ and $\widetilde{\mathbf{U}_1}=\begin{bmatrix}1 & 0 & 0\end{bmatrix}$, $\widetilde{\mathbf{U}_2}=\begin{bmatrix}0 & 1 & 0\end{bmatrix}$, $\widetilde{\mathbf{U}_3}=\begin{bmatrix}0 & 0 & 1\end{bmatrix}$, each projection $\widetilde\pi_1^i$ can be represented by the projection matrix: $$\widetilde{\mathbf P}_1^i=\widetilde{\mathbf{M}}_i\widetilde{\mathbf{U}}_i(\widetilde{\mathbf{M}}^T\widetilde{\mathbf{M}})^{-1}\widetilde{\mathbf{M}}^T,~i=1,2,3.$$ By definition, $\mathbf{Y}_1^j=\widetilde{\mathbf P}_1^j\mathbf{Y}_1$, $\mathbf{E}_1^j=\widetilde{\mathbf P}_1^j\mathbf{E}_1$ and $\mathcal{O}_1^j=\widetilde{\mathbf P}_1^j\mathcal{O}_1\vert_{V^j}$ for $j=1,2,3$. In summary, the decomposition of measurement $\mathbf{Y}_1(k)=\mathcal{O}_1\mathbf{x}(k)+\mathbf{E}_1(k)$ is given by: $$\widetilde{\mathbf P}_1^j\mathbf{Y}_1(k)=(\widetilde{\mathbf P}_1^j\mathcal{O}_1\vert_{V^j})(\mathbf{P}_1^j\mathbf{x}(k))+\widetilde{\mathbf{P}}_1^j\mathbf{E}_1(k), \quad j=1,2,3.$$

\section{CLASSES OF SSR PROBLEMS SOLVABLE IN POLYNOMIAL TIME}
\label{sec:Theorem1}
While in the previous section we established that the SSR problem is NP-hard, in this section we leverage the results in Section \ref{sec:decompose} to answer a simple but important question: when can we solve the SSR problem in polynomial time? Our answer relies heavily on the system decomposition technique introduced in Section \ref{sec:decompose}. The first result establishes that the decomposition can be done in polynomial time.

\begin{proposition}
\label{prop:com}
The computational complexity of decomposing the system \eqref{eqn:modelsys}-\eqref{eqn:modelmeasure} into sub-systems \eqref{eqn:decompose} is within $O(pn^3)$.
\end{proposition}

\begin{proof}
To prove this result, we list all the steps involved in the decomposition from \eqref{eqn:modelsys}-\eqref{eqn:modelmeasure} to \eqref{eqn:decompose} and list the computational complexity of each step.

	\textbf{\underline{Offline preparation 1}}: compute the observability matrix of each sensor $\mathcal{O}_i$. The computational complexity of this step is $O(pn^2)$.
	
	\textbf{\underline{Offline preparation 2}}: find the eigenvalues of the matrix $\mathbf{A}$ as well as its generalized eigenspaces $V^j$. This can be done by finding the Jordan form of $\mathbf{A}$. The computational complexity of this step is $O(n^3)$.
	
	\textbf{\underline{Offline preparation 3}}: determine the image of each generalized eigenspace $V^j$ under the observability matrix $\mathcal{O}_i$, i.e., $\mathcal{O}_i(V^j)$. In this step, we perform $p$ times two $n\times n$ matrix multiplications and thus the complexity of this step is $O(pn^3)$.
	
	\textbf{\underline{Offline preparation 4}}: find the projection matrix for each generalized eigenspace and each sensor. The computational complexity of this step is $O(pn^3)$.
	
	\textbf{\underline{Online task}}: at each time instance, project the measurements $\mathbf{Y}_i(k)$ of each sensor $i$ onto each generalized eigenspace. In this step, for each sensor we multiply a $n\times n$ matrix by a $n\times 1$ vector $r$ times. This requires $O(pn^2r)$ time.  
	
We thus conclude that we can decompose the system \eqref{eqn:modelsys}-\eqref{eqn:modelmeasure} into sub-systems \eqref{eqn:decompose} within $O(pn^3)$ and finish the proof.
\end{proof}


Before giving an answer to the question we stated at the beginning of this section, we relate the sparse observability index defined for the system \eqref{eqn:modelsys}-\eqref{eqn:modelmeasure} and the sparse observability index for each subsystem \eqref{eqn:decompose} with $j$ ranging from $1$ to $r$ in the following two results. Note that, since the state space of \eqref{eqn:decompose} is $V^j$, sparse observability is characterized by the injectivity of $\mathcal{O}_i^j\vert_{V^j}$ whereas eigenvalue observability is characterized by injectivity of the linear map $\begin{bmatrix}\mathbf{A}^{(j)}-\lambda_j \mathbf{I}_n^{(j)}\\\mathbf{C}_i^j\end{bmatrix}$, where we define $\mathbf{I}_n^{(j)}=\pi_j\circ \mathbf{I}_n\circ\imath_j$. We now have the following results.

\begin{theorem}
\label{thm:equiv1}
The system \eqref{eqn:modelsys}-\eqref{eqn:modelmeasure} is $k$-sparse observable if and only if for each $j \in \{1, 2, \ldots, r\}$, the system \eqref{eqn:decompose} is $k$-sparse observable.
\end{theorem}
\begin{proof}
This result can be easily established by observing that $\mathrm{ker}~\mathcal{O}_i=\oplus_{j=1}^r\mathrm{ker}~\mathcal{O}_i^j$ holds for any sensor $i$. We omit the proof here in the interest of space.
\end{proof}

Similarly, to relate the eigenvalue observability index defined for the overall system and the eigenvalue observability index for each subsystem, we have the following result.

\begin{theorem}
\label{thm:equiv2}
The system \eqref{eqn:modelsys}-\eqref{eqn:modelmeasure} is $k$-eigenvalue observable if and only if for each $j \in \{1, 2, \ldots, r\}$, the system \eqref{eqn:decompose} is $k$-eigenvalue observable.
\end{theorem}

\begin{proof}
By the definition of eigenvalue observability, it suffices to show the matrix $\begin{bmatrix}\mathbf{A}-\lambda_j\mathbf{I}_n \\ \mathbf{C}_i\end{bmatrix}$ has full column rank if and only if each matrix $\begin{bmatrix}\mathbf{A}^{(j)}-\lambda_j\mathbf{I}_n^{(j)} \\ \mathbf{C}_i^j\end{bmatrix}$ defines an injective map with domain $V^j$, for $j$ ranging from $1$ to $r$.

Consider the map $F:V\rightarrow V\times \mathbb{R}^{p_i}$ defined by the matrix $\begin{bmatrix}\mathbf{A}-\lambda_j\mathbf{I}_n \\ \mathbf{C}_i\end{bmatrix}$ and note that $F$ being injective is equivalent to $\ker F=\{0\}$. Note also that the result immediately follows if we establish that $\ker F\subseteq V^{j}$. This can be seen by noting that $F\mathbf{x}=0$ for $\mathbf{x}\in \mathbb{R}^n$ degenerates to $F\mathbf{x}=0$ for $\mathbf{x}\in V^j$ and (given $\mathbf{x}=\imath_j \mathbf{x}$) can be written as $F\imath_j \mathbf{x}=0$:
	\begin{equation}
	    \label{Eq:MatrixEq1}
	    \begin{bmatrix}\mathbf{A}\circ \imath_j-\lambda_j\imath_j \\ \mathbf{C}_i\circ  \imath_j\end{bmatrix}\mathbf{x}=0.
	\end{equation}
	Moreover, since $(\mathbf{A}-\lambda_j \mathbf{I}_n)(V^j)\subseteq V^j$ we have the equality $\pi_j(\mathbf{A}-\lambda_j \mathbf{I}_n)\imath_j \mathbf{x}=(\mathbf{A}-\lambda_j \mathbf{I}_n)\imath_j \mathbf{x}$.  Therefore,~\eqref{Eq:MatrixEq1} degenerates into:
		\begin{equation}
	    \begin{bmatrix}\pi_j\circ \mathbf{A}\circ \imath_j-\lambda_j\pi_j\circ\imath_j \\ \mathbf{C}_i\circ  \imath_j\end{bmatrix}\mathbf{x}=
	    \begin{bmatrix} \mathbf{A}^{(j)}-\lambda_j \mathbf{I}_n^{(j)} \\ \mathbf{C}_i^j\end{bmatrix}\mathbf{x}=0.
	\end{equation}
	Therefore, we proceed by showing that $\ker F\subseteq V^j$. The equality $F\mathbf{x}=0$ implies $(\mathbf{A}-\lambda_j \mathbf{I}_n)\mathbf{x}=0$. If we write $\mathbf{x}$ as $\mathbf{x}_j+\mathbf{x}_{\overline{j}}$ with $\mathbf{x}_j=\pi_j(\mathbf{x})$ and $\mathbf{x}_{\overline{j}}=\sum_{k=1, k\ne j}^r \pi_k(\mathbf{x})$ we have $(\mathbf{A}-\lambda_j \mathbf{I}_n)(\mathbf{x}_j+\mathbf{x}_{\overline{j}})=0$. We now make two observations. The first is that $(\mathbf{A}-\lambda_j \mathbf{I}_n)\mathbf{x}_{\overline{j}}= 0$ implies $\mathbf{x}_{\overline{j}}=0$ since $\mathbf{x}_{\overline{j}}\ne 0$ would imply that $\mathbf{x}_{\overline{j}}\in V^j$, by definition of $V^j$. The second observation is that $(\mathbf{A}-\lambda_j \mathbf{I}_n)(V^\ell)\subseteq V^\ell$, for $\ell\in \{1,\hdots,r\}$, implies that $(\mathbf{A}-\lambda_j \mathbf{I}_n)(\mathbf{x}_j+\mathbf{x}_{\overline{j}})=0$ iff $(\mathbf{A}-\lambda_j \mathbf{I}_n)\mathbf{x}_j=0$ and $(\mathbf{A}-\lambda_j \mathbf{I}_n)\mathbf{x}_{\overline{j}}=0$. Together with the first observation we have $\mathbf{x}_{\overline{j}}=0$ which implies that $\mathbf{x}\in V^j$ and concludes the proof.
\end{proof}

Based on the above decomposition and the assumption that at most $s$ sensors are attacked, we partition the set of eigenvalues $\{\lambda_1, \lambda_2, \ldots, \lambda_r\}$ as follows:
\begin{itemize}
    \item We define $\mathcal{J}_1 \subseteq \{\lambda_1, \lambda_2, \ldots, \lambda_r\}$ to be the set of eigenvalues whose corresponding subsystems \eqref{eqn:decompose} are not $2s$-sparse observable.
    \item We define $\mathcal{J}_2 \subseteq \{\lambda_1, \lambda_2, \ldots, \lambda_r\}\setminus \mathcal{J}_1$ to be the set of eigenvalues whose corresponding subsystems \eqref{eqn:decompose} are $2s$-eigenvalue observable.
    \item We define $\mathcal{J}_3 = \{\lambda_1, \lambda_2, \ldots, \lambda_r\}\setminus \{\mathcal{J}_1 \cup \mathcal{J}_2\}$ to be the set of eigenvalues whose corresponding subsystems \eqref{eqn:decompose} are $2s$-sparse observable but not $2s-$eigenvalue observable.
\end{itemize}

\subsection{Impossibility of Reconstructing Substates Corresponding to Eigenvalues in the Set $\mathcal{J}_1$}
It is established in Section \eqref{sec:prob_form} that the SSR problem does not admit a unique solution if it is not $2s-$sparse observable. Therefore, it is impossible to reconstruct the substates corresponding to eigenvalues in $\mathcal{J}_1$. Furthermore, by Theorem~\eqref{thm:equiv1} if $\mathcal{J}_1$ is not empty, the overall system defined in \eqref{eqn:modelsys}-\eqref{eqn:modelmeasure} is not $2s-$sparse observable, which in turn means the solution is not unique.

\subsection{Reconstructing the Substates Corresponding to Eigenvalues in the Set $\mathcal{J}_2$}

We learned from Theorem~\eqref{thm:equiv2} that if $\lambda_j$ is observable w.r.t. sensor $i$, then after decomposing the system, $\lambda_j$ is also observable w.r.t. to sensor $i$ in the $j$-th sub-system corresponding to this sensor. By the Popov-Belevitch-Hautus (PBH) test, the $j$-th sub-system ($\mathbf{A}^{(j)},\mathbf{C}_i^j$) is observable, which shows that $\mathbf{x}_j$ can be reconstructed using only measurements from sensor $i$.

We now explain how to reconstruct the substates corresponding to eigenvalues in $\mathcal{J}_2$ based on majority voting. Consider any eigenvalue $\lambda_j\in \mathcal{J}_2$. Let $S_{\lambda_j}$ represent the set of sensors w.r.t. which $\lambda_j$ is observable. The result of the PBH test implies that $\mathbf{x}_j$ can be recovered using the measurements of each of the sensors in the set $\mathcal{S}_{\lambda_j}$. We denote by $x_j^{(l)}$ the $l$th component of $\mathbf{x}_j$. Based on the definition of the set $\mathcal{J}_2$, we have $|\mathcal{S}_{\lambda_j}| \geq (2s+1)$. Consequently, since at most $s$ sensors have been compromised, we are guaranteed at least $s+1$ consistent copies of the state $x_j^{(l)}$. Thus, each component of the vector $x_j^{(l)}$ can be recovered via majority voting and therefore all the substates corresponding to eigenvalues in $\mathcal{J}_2$ can be reconstructed in polynomial time.

\subsection{Computational Complexity of Reconstructing Substates Corresponding to Eigenvalues in the Set $\mathcal{J}_3$}

The NP-hardness of solving the SSR problem has been established in Section~\ref{sec: SSRhard}. In this subsection, we argue that with the prescribed decomposition technique, the computational complexity of solving the SSR problem for substates corresponding to eigenvalues in $\mathcal{J}_3$ could be reduced whenever we only need to reconstruct substates whose dimension is smaller than $n$. Assuming $s$ is the upper bound of the number of attacked sensors, we have the following theorem.

\begin{theorem}
	\label{thm:reduction}
	By applying the decomposition \eqref{eqn:decompose}, the SSR problem can be solved in time $\sum_{\lambda_j\in\mathcal{J}_3}\mathcal{C}(p,n_j)+O(pn^3)$ if the system \eqref{eqn:modelsys}-\eqref{eqn:modelmeasure} is $2s-$sparse observable, where $\mathcal{C}(p,n)$ is the time complexity of solving an instance of the SSR problem with $n$ states and $p$ sensors whose corresponding system is $2s-$sparse observable.
\end{theorem}

Before providing a proof we first discuss how this result may reduce the computational complexity of solving the SSR problem. For a large-scale CPS, it's not uncommon for the number of sensors to greatly exceed the number of states, i.e., $p\gg n$. We note that the computational complexity of brute force search grows exponentially with $p$. Also, the computational complexity of some brute force search algorithms (such as~\cite{7171098}) to determine whether a set of sensors is attacked is at least $O(n^2)$ . In other words, for such algorithms $\mathcal{C}(p,n)\geq O(p^2n^2)$. By assuming $p\gg n$ we make the following observations:
\begin{enumerate}
	\item $O(p^2n^2)\geq \sum_{j=1}^r O(p^2n_j^2)$, and equality holds only when $r=1$.
	\item $O(pn^3)\ll \sum_{j=1}^r O(p^2n_j^2)$.
\end{enumerate}

The first observation shows that the computation required to solve all the sub-problems is smaller than what is required to solve the original problem. The second observation shows that, compared with the computational complexity of solving the SSR problem, the computation required for decomposition of the original system is negligible. These two facts indicate that by decomposing the SSR problem into simpler instances, we reduce the computational complexity of solving the SSR problem.
\vspace{1.8mm}

\noindent \textit{Proof of Theorem~\ref{thm:reduction}}:
	We already established that reconstructing the state of each decomposed system is also an SSR problem and the solution $\mathbf{x}$ of the original problem is obtained by summing over all the projections, i.e., $\mathbf{x}=\mathbf{x}_1+\mathbf{x}_2+\dots +\mathbf{x}_r$. Therefore any algorithm that solves the SSR problem can be applied to solve each subproblem, i.e., we may solve each subproblem corresponding to $\lambda_j\in\mathcal{J}_3$ within time complexity $\mathcal{C}(p,n_j)$ since there are $p$ sensors and $n_j$ states. By the assumption that the system \eqref{eqn:modelsys}-\eqref{eqn:modelmeasure} is $2s-$sparse observable as well as Theorem~\eqref{thm:equiv1}, all sub-systems are $2s-$sparse observable and hence $\mathcal{J}_1=\{\phi\}$, and for each subproblem corresponding to $\lambda_j\in\mathcal{J}_2$ the time complexity of the majority voting algorithm is within $O(pn^2)$. In summary, the total computational complexity is: 
\begin{eqnarray}
&&\sum_{\lambda_j\in\mathcal{J}_2}O(pn_j^2)+\sum_{\lambda_j\in\mathcal{J}_3}\mathcal{C}(p,n_j)+O(pn^3) \\ 
=&&\sum_{\lambda_j\in\mathcal{J}_3}\mathcal{C}(p,n_j)+O(pn^3), 
\end{eqnarray}
which finishes the proof.

\hfill \qedsymbol

\begin{remark}
The actual complexity might be even smaller than $\sum_{\lambda_j\in\mathcal{J}_3}\mathcal{C}(p,n_j)+O(pn^3)$. This can be seen by noting that we solve each smaller SSR problem sequentially, and thus we can remove measurements from sensors that have been identified as being attacked when solving subsequent problems.
\end{remark}

To conclude, we have the following result which answers the question at the beginning of this section by pointing out when the SSR problem can be solved in polynomial time, which actually is a corollary of Theorem~\eqref{thm:reduction}.
\begin{corollary} 
\label{the:poly}
Consider the system \eqref{eqn:modelsys}-\eqref{eqn:modelmeasure}, and suppose at most $s$ sensors are attacked. Let the eigenvalue observability index of system \eqref{eqn:modelsys}-\eqref{eqn:modelmeasure} be at least $2s$. Then, the SSR problem can be solved in polynomial time. 
\end{corollary}


\begin{remark}
Another understanding of this classification of eigenvalues into $\mathcal{J}_1$, $\mathcal{J}_2,$ and $\mathcal{J}_3$ is provided by the vulnerability of the corresponding substates. Substates in $\mathcal{J}_1$ are the most vulnerable to attack since the defender may not even be able to identify the attacked set of sensors. Substates in $\mathcal{J}_2$ are robust against attacks since attacked sensors can be easily determined. For substates $\mathcal{J}_3$, the defender is able to identify the attacked sensors, but this task requires a substantially higher computational effort.

In other words, in the view of the adversary, a wise attacking strategy is to attack the substates corresponding to eigenvalues in $\mathcal{J}_1$, and it should avoid attacking states in $\mathcal{J}_2$ since majority voting will allow the defender to easily identify the compromised sensors.
\end{remark}

\subsection{Example - Continued}
In this subsection we continue the example in Section~\ref{sec:prelim} and Section~\ref{sec:decompose} and show how to classify each subsystem under the assumption that the adversary can attack at most $s=1$ sensor. We recall that $V^1,V^2,V^3$ are the eigenspaces corresponding to eigenvalues 1, 2, and 3, respectively. Also, after decomposition, we have  $\mathbf{A}^{(j)}=\mathbf{P}_j\mathbf{A}\vert_{V^j}$ as well as $\mathcal{O}_i^j=\widetilde{\mathbf{P}}_i^j\mathcal{O}_i\vert_{V^j}$ for $i=1,2,3,4$ and $j=1,2,3.$ 

We first claim that $\lambda_3=3$ belongs to $\mathcal{J}_1$. To see why this is true, we remove $2s=2$ sensors, sensor 1 and sensor 4, and explicitly compute $\mathcal{O}_2^3$ and $\mathcal{O}_3^3$. We have: $$\mathcal{O}_2=\begin{bmatrix}2 & 3 & \phantom{-}1 & -1\\ 3 & 4 & \phantom{-}0 & -1 \\ 5 & 6 & -2 & -1 \\ 9 & 10 & -6 & -1\end{bmatrix},\mathcal{O}_3=\begin{bmatrix}2 & 2 & \phantom{-}0 & \phantom{-}0\\ 3 & 3 & -1 & \phantom{-}0 \\ 5 & 5 & -3 & -0 \\ 9 & 9 & -7 & \phantom{-}0\end{bmatrix},$$ and $\mathcal{O}_2(V^3)=\mathcal{O}_3(V^3)=\{0\}$ which yields $(\widetilde{\mathbf{P}}_2^3\mathcal{O}_2)\mathbf{x}'_3=0$ and $(\widetilde{\mathbf{P}}_3^3\mathcal{O}_3)\mathbf{x}''_3=0$ for any $\mathbf{x}'_3$ and $\mathbf{x}''_3$ in $V^3$. Therefore, we have $\mathcal{O}_2^3=\mathcal{O}_3^3=0.$ By the definition of sparse observability, we have $\mathrm{ker}~\mathcal{O}_{\{2,3\}}^{3}=V^3$ and hence the subsystems corresponding to eigenvalue $3$ are not $2s-$sparse observable. Also, a similar analysis reveals that subsystems corresponding to eigenvalues $\lambda_1$ and $\lambda_2$ are both $2s-$sparse observable, hence $1\notin \mathcal{J}_1$ and $2\notin \mathcal{J}_1$.

Next we argue that $\lambda_2=2$ belongs to $\mathcal{J}_2$. To see why this is true, we first recall that $\mathbf{A}^{(2)}=\mathbf{P}_2\mathbf{A}\vert_{V^2}$, $\mathbf{I}_4^{(2)}=\mathbf{I}_4\vert_{V^2}$, $\mathbf{C}_i^2=\mathbf{C}_i\vert_{V^2}$, and then check that for sensor 1, the matrix: $$\begin{bmatrix}\mathbf{A}^{(2)}-2\mathbf{I}_4^{(2)} \\ \mathbf{C}_1^2 \end{bmatrix} = \left.\begin{bmatrix} -2 & \phantom{-}0 & \phantom{-}0 & \phantom{-}0 \\ \phantom{-}1 & -1 & -1 & \phantom{-}0 \\ -1 & -1 & -1 & \phantom{-}0 \\ \phantom{-}0 & \phantom{-}0 & \phantom{-}0 & -2 \\ \phantom{-}3 & \phantom{-}2 & \phantom{-}0 & \phantom{-}2
\end{bmatrix}\right\vert_{V^2},$$ defines an injective map. We also run the same check on sensor $2,3,$ and $4$ to conclude that eigenvalue $\lambda_2$ is observable by all $4$ sensors. Hence the subsystems corresponding to $\lambda_2$ are $2s-$eigenvalue observable. Proceeding in the same fashion we conclude that subsystems corresponding to eigenvalue $\lambda_1$ are not $2s-$eigenvalue observable. Therefore, the eigenvalue $\lambda_1=1$ belongs to $\mathcal{J}_3$.

In summary, the substates in $V^3$ cannot be securely reconstructed, the substates in $V^1$ can be securely reconstructed in the presence of at most $1$ attacked sensor, and the substates in $V^2$ can be securely reconstructed and the reconstruction can be done efficiently.

\section{COMPLEXITY OF CHECKING SPARSE OBSERVABILITY}
\label{sec:theorem2}
In the previous two sections, we studied the complexity of the SSR problem, and in particular, identified instances of the problem that can be solved in polynomial time. Recall that under at most $s$ sensor attacks on the system ~\eqref{eqn:modelsys}-\eqref{eqn:modelmeasure}, $2s$-sparse observability is necessary and sufficient for the SSR problem to yield a unique solution, namely the true initial state vector $\mathbf{x}(0)$. Given this result, we now take a step back and ask: what is the complexity of deciding whether a given system is $2s$-sparse observable? This question is highly relevant since it aims to identify the maximum number of sensor attacks that can be tolerated by a given system of the form \eqref{eqn:modelsys}-\eqref{eqn:modelmeasure}. In what follows, we show that determining the sparse-observability index (see Definition \ref{def:sparse_obs_index}) of a system is computationally hard; we will focus on the case of scalar-valued sensors throughout, as it suffices to establish the computational complexity of the problem.

\begin{problem} ($r$-sparse observability)\\ {\bf Input:} A matrix $\mathbf{A} \in \mathbb{Q}^{n\times n}$, a matrix $\mathbf{C} \in \mathbb{Q}^{p \times n}$ and a positive integer $r$.\\
{\bf Question:} Is the pair $(\mathbf{A}, \mathbf{C})$ $r$-sparse observable?
\end{problem}

Note that if the answer to an instance of the $r$-sparse observability problem is ``no'', then there is a simple proof: one can provide a set of $r$ rows of $\mathbf{C}$ that, if removed, result in a system that is no longer observable.  However, it is not clear whether there is a similarly simple proof for ``yes'' instances.  Thus, the $r$-sparse observability problem is in the class coNP.\footnote{See, e.g., \cite{cormen2009introduction} for additional details on the complexity classes NP and coNP.}

The complement of a decision problem is the problem obtained by switching the ``yes'' and ``no'' answers to all instances of that problem.  If a problem is in the class coNP, then its complement is in the class NP, and vice versa.

We will show that the $r$-sparse observability problem is coNP-hard by showing that its complement is NP-hard.  Specifically, we define the following complement problem to $r$-sparse observability.

\begin{problem} ($r$-sparse unobservability)\\
{\bf Input: } A matrix $\mathbf{A} \in \mathbb{Q}^{n\times n}$, a matrix $\mathbf{C} \in \mathbb{Q}^{p \times n}$ and a positive integer $r$.\\
{\bf Question:} Is there a set of $r$ rows that can be removed from $\mathbf{C}$ in order to yield a matrix $\bar{\mathbf{C}}$ such that $(\mathbf{A},\bar{\mathbf{C}})$ is unobservable?
\end{problem}
 
Note that the answer to an instance of $r$-sparse unobservability is ``yes'' if and only if the answer to the corresponding instance of $r$-sparse observability is ``no'' and vice versa. Further note that $r$-sparse unobservability is in the class NP.

We show that $r$-sparse unobservability is NP-complete by providing a reduction from the following {\it Linear Degeneracy} problem.  This problem was shown to be NP-complete in \cite{khachiyan1995complexity}.

\begin{problem} (Linear Degeneracy~\cite{khachiyan1995complexity})\\
{\bf Input:} A full column rank matrix $\mathbf{F} \in \mathbb{Q}^{p \times n}$.  \\
{\bf Question:} Does $\mathbf{F}$ contain a degenerate (i.e., noninvertible) $n \times n$ submatrix?
\end{problem}

In other words, the linear degeneracy problem asks whether it is possible to remove $p-n$ rows from matrix $\mathbf{F}$ so that the resulting (square) matrix is not full rank.  We are now ready to prove  the following result.

\begin{theorem}[\cite{Yanwen}]
The $r$-sparse unobservability problem is NP-complete.  Thus, the $r$-sparse observability problem is coNP-complete.
\label{thm:checkhard}
\end{theorem}

\begin{proof}
Given an instance of the linear degeneracy problem (with matrix $\mathbf{F} \in \mathbb{Q}^{p\times n}$), we construct an instance of the $r$-sparse unobservability problem as follows: set $\mathbf{A} = \mathbf{I}_n$, $\mathbf{C} = \mathbf{F}$, and $r = p-n$.

We now show that the answer to the constructed instance of $r$-sparse unobservability is ``yes'' if and only if the answer to the given instance of linear degeneracy is ``yes''.

First, suppose that the answer to the constructed instance of $r$-sparse unobservability is ``yes.''  Then there exists a set of $r$ rows of $\mathbf{C}$ that can be removed such that the remaining rows are not sufficient to yield observability.  However, since $\mathbf{A} = \mathbf{I}_n$, the above implies that there is a set of $r$ rows of $\mathbf{C}$ that can be removed such that the remaining rows are not full column rank.  Since $\mathbf{C} = \mathbf{F}$ and $r = p-n$, this means that there is an $n \times n$ submatrix of $\mathbf{F}$ that loses rank, and thus the answer to the linear degeneracy problem is ``yes.''

Next, we show that if the answer to the given instance of linear degeneracy is ``yes,'' then the answer to the constructed instance of $r$-sparse unobservability is ``yes.'' We will do this by showing the contrapositive: if the answer to the constructed instance of $r$-sparse unobservability is ``no'', then the answer to the given instance of linear degeneracy is ``no.''  Suppose the answer to the constructed instance of $r$-sparse unobservability is ``no.'' Then, by definition, the pair $(\mathbf{A},\mathbf{C})$ is observable even after removing any arbitrary $r$ rows from $\mathbf{C}$.  However, since $\mathbf{A} = \mathbf{I}_n$, in order for the system to remain observable after removing $r$ rows from $\mathbf{C}$, it must be the case that the remaining rows of $\mathbf{C}$ have full column rank.  Thus, if the answer to the constructed instance of $r$-sparse unobserability is ``no'', then $\mathbf{C}$ has full column rank after removing any arbitrary $r= p-n$ rows.  This means that every $n \times n$ submatrix of $\mathbf{C}$ is invertible.  Since $\mathbf{C} = \mathbf{F}$, the answer to the given instance of linear degeneracy is ``no'' (i.e., there is no $n \times n$ submatrix of $\mathbf{F}$ that is degenerate).

Thus, we have shown that the answer to the constructed instance of $r$-sparse unobservability is ``yes'' if and only if the answer to the given instance of linear degeneracy is ``yes". Since linear degeneracy is NP-complete, so is $r$-sparse unobservability.

Finally, since $r$-sparse observability is the complement of $r$-sparse unobservability, we have that $r$-sparse observability is coNP-complete.
\end{proof}

\begin{remark} 
In \cite{mitraAuto}, certain necessary conditions were presented for estimating the state of a plant despite attacks in a distributed setting, i.e., where measurements of the plant are dispersed over a network of sensors. Specifically, these conditions impose certain requirements on the observation model (in addition to requirements on the communication structure), the complexity of checking which was left open. Interestingly, Theorem \ref{thm:checkhard} resolves this question, and establishes that checking the necessary conditions in \cite{mitraAuto} is computationally hard; since the focus of our paper is on centralized systems, we do not present details of this result here. 
\end{remark}

\section{CONNECTIONS BETWEEN SPARSE OBSERVABILITY AND EIGENVALUE OBSERVABILITY}
\label{sec:connection}
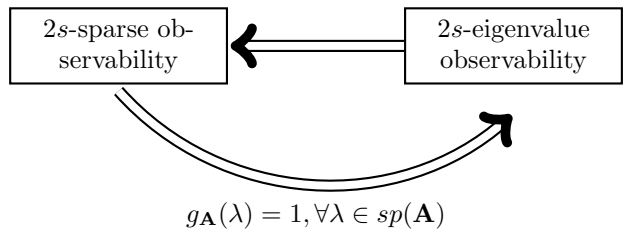
\begin{figure}[t]
\begin{center}
\begin{tikzpicture}
[->,shorten >=1pt,scale=.75,inner sep=1pt, minimum size=12pt, auto=center, node distance=3cm,
  thick, node/.style={circle, draw=black, thick},]
\tikzstyle{block1} = [rectangle, draw, fill=white, 
    text width=8em, text centered, minimum height=1cm, minimum width=1cm];
 \node [block1]  at (0,0) (c1) {$2s$-sparse observability};
\node  [block1]  at (7,0) (c2)
{$2s$-eigenvalue observability};
\draw[double,double distance=3pt, ->](c2) -- (c1);
\draw[double,double distance=3pt, ->](0,-0.8) to [bend right=45]  (7,-1.2);
\node at (3.5,-3) {$g_{\mathbf{A}}(\lambda)=1,\forall \lambda\in sp(\mathbf{A})$};
\end{tikzpicture}
\end{center}
\caption{Figure illustrating the hierarchy of relationships between different notions of observability.}
\label{fig:hierarchy}
\end{figure}
\label{sec:generalization}
In Sections \ref{sec: SSRhard} and \ref{sec:theorem2}, we showed that the SSR problem and the problem of determining the sparse observability index of a system are each computationally hard. At the same time, Section \ref{sec:Theorem1} gave us the positive result that certain instances of the SSR problem can be efficiently solved. In line with this finding, we are now motivated to ask: Can the sparse observability index of a system be computed in polynomial time for certain specific instances? In this section, we show that this is indeed the case by identifying instances of the problem where the notions of sparse observability and eigenvalue observability coincide. Given that the eigenvalue observability index of a system can always be computed in polynomial time based on simple rank tests, an equivalence between the two notions of observability immediately yields instances of the problem where the sparse observability index of the system can also be computed in polynomial time. With this in mind, in this section we will prove each of the implications indicated in Figure \ref{fig:hierarchy}. We begin with the following simple result.

\begin{proposition}[\cite{Yanwen}]
	Consider the linear system~\eqref{eqn:modelsys}-\eqref{eqn:modelmeasure}, and suppose its eigenvalue observability index is $2s$. Then, the pair $(\mathbf{A},\mathbf{C})$ is at least $2s$-sparse observable.
	\label{prop:implication1}
\end{proposition}

\begin{proof}
	Consider any subset of sensors $\mathcal{F}\subset\mathcal{V}$, such that $|\mathcal{F}|\leq 2s$. To establish that the pair $(\mathbf{A},\mathbf{C})$ is at least $2s$-sparse observable, we need to show that the pair $(\mathbf{A},\mathbf{C}_{\mathcal{V}\setminus\mathcal{F}})$ is observable. Based on the PBH test, this amounts to checking that each eigenvalue $\lambda\in sp(\mathbf{A})$ is observable w.r.t. the  observation matrix $\mathbf{C}_{\mathcal{V}\setminus\mathcal{F}}$. Let $S_{\lambda}$ represent the set of sensors w.r.t. which $\lambda$ is observable. A sufficient condition for this to happen is $|(\mathcal{V}\setminus\mathcal{F})\cap\mathcal{S}_{\lambda}|\geq1$, which is indeed true given that an eigenvalue observability index of $2s$ implies $|\mathcal{S}_{\lambda}|\geq (2f+1), \forall \lambda \in sp(\mathbf{A})$, and the fact that $|\mathcal{F}| \leq 2s$. 
\end{proof}

To see that the reverse implication does not hold in general, consider the following example.

\begin{example}
Consider an LTI system of the form \eqref{eqn:modelsys}-\eqref{eqn:modelmeasure} monitored by 6 sensors, with parameters as follows:
\begin{equation}
    \mathbf{A}=\begin{bmatrix}\lambda&0\\0&\lambda\end{bmatrix}, 
    \mathbf{C}_i=
    \begin{cases}
    \begin{bmatrix}1&0\end{bmatrix},& \text{if } i\in\{1,2,3\},\\ \\
    \begin{bmatrix}0&1\end{bmatrix},              & \text{if } i\in\{4,5,6\}.
\end{cases}
\label{eqn:counter}
\end{equation}
Here $\lambda\in\mathbb{R},|\lambda|\geq1$. Suppose $s=1$. Then, the removal of at most $2$ sensors will ensure that at least one sensor from each of the sets $\{1,2,3\}$ and $\{4,5,6\}$ remains unattacked; given the measurement model in \eqref{eqn:counter}, this is sufficient to preserve observability w.r.t. the remaining sensors. In other words, the system is 2-sparse observable. However, it is easy to verify that the eigenvalue $\lambda$ is not observable w.r.t. any sensor. 
\label{ex:example}
\end{example}

In view of Proposition \ref{prop:implication1} and Example \ref{ex:example}, we conclude that  $2s$-sparse observability of a system is in general less restrictive than the condition that the eigenvalue observability index of the system is $2s$. In what follows, we establish that the two aforementioned notions coincide when additional structure is imposed on the spectrum of $\mathbf{A}$. 

\begin{proposition}[\cite{Yanwen}] Consider the linear system model given by  \eqref{eqn:modelsys}-\eqref{eqn:modelmeasure}, and suppose $\lambda \in sp(\mathbf{A})$ has geometric multiplicity $1$. Consider any non-empty subset of sensors $\mathcal{S}=\{i_1,i_2,\ldots,i_{|\mathcal{S}|}\} \subseteq\mathcal{V}$. Then, the eigenvalue $\lambda$ is observable w.r.t. the pair $(\mathbf{A},\mathbf{C}_{\mathcal{S}})$ if and only if there exists a sensor $i_p\in\mathcal{S}$ such that $\lambda$ is observable w.r.t. sensor $i_p$, i.e., $\lambda$ is observable w.r.t. the pair $(\mathbf{A},\mathbf{C}_{i_p})$.
\label{prop:geommult1}
\end{proposition}

\begin{proof}
Consider a similarity transformation that maps $\mathbf{A}$ to its Jordan canonical form $\mathbf{J}$. Let this transformation map $\mathbf{C}_{\mathcal{S}}$ to $\bar{\mathbf{C}}_{\mathcal{S}}$, and ${\mathbf{C}}_{i_j}$ to $\bar{\mathbf{C}}_{i_j}$, for each $i_j\in\mathcal{S}$. Since $\lambda$ has geometric multiplicity $1$, there exists a single Jordan block corresponding to $\lambda$ in $\mathbf{J}$. Let this Jordan block be denoted $\mathbf{J}_{\lambda}$. Without loss of generality, suppose $\mathbf{J}$ is of the following form:
\begin{equation}
\mathbf{J}=\begin{bmatrix} \mathbf{J}_{\lambda}&\mathbf{0}\\ \mathbf{0}&\bar{\mathbf{J}}\end{bmatrix},
\label{eqn:jordanstruct}
\end{equation}
where $\bar{\mathbf{J}}$ is the collection of the Jordan blocks corresponding to eigenvalues in $sp(\mathbf{A})\backslash \{\lambda\}$. 
Based on the PBH test, $\lambda$ is observable w.r.t. the pair $(\mathbf{J},\bar{\mathbf{C}}_{\mathcal{S}})$ if and only if the following condition holds:
\begin{equation}
\label{eqn:cond1}
\textrm{rank}\begin{bmatrix}\mathbf{J}-\lambda\mathbf{I}_n\\ \bar{\mathbf{C}}_{\mathcal{S}}\end{bmatrix}=n.
\end{equation}
Given the structure of $\mathbf{J}$ in \eqref{eqn:jordanstruct}, and the fact that $\lambda$ has geometric multiplicity $1$, it is easy to see that \eqref{eqn:cond1} holds if and only if there is at least one non-zero entry in the first column of $\bar{\mathbf{C}}_{\mathcal{S}}$. However, the preceding condition holds if and only if there exists some sensor $i_p\in\mathcal{S}$ with at least one non-zero entry in the first column of $\bar{\mathbf{C}}_{i_p}$; the latter is precisely the condition for observability of $\lambda$ w.r.t. the sensor $i_p$, given that $g_{\mathbf{A}}(\lambda)=1.$ To complete the proof, it suffices to notice that a similarity transformation preserves the observability of an eigenvalue. 
\end{proof}

We now make use of the previous result to establish an equivalence between sparse observability and eigenvalue observability.
\begin{proposition}
\label{PropEquivalence}Consider the linear system model~\eqref{eqn:modelsys}-\eqref{eqn:modelmeasure}, and suppose every eigenvalue of $\mathbf{A}$ has geometric multiplicity $1$. Then, the pair $(\mathbf{A},\mathbf{C})$ is $2s$-sparse observable if and only if the eigenvalue observability of the system is $2s$.
\label{prop:connection}
\end{proposition}
\begin{proof}
For necessity, we proceed via contradiction. Suppose the pair $(\mathbf{A},\mathbf{C})$ is $2s$-sparse observable, but there exists some $\lambda \in sp(\mathbf{A})$ that is observable w.r.t. at most $2s$ distinct sensors. Recall that the set of sensors w.r.t. which $\lambda$ is observable is denoted $\mathcal{S}_{\lambda}$. Based on our hypothesis, $|\mathcal{S}_{\lambda}| \leq 2s$. Suppose $|\mathcal{S}_{\lambda}|=2s$ (since an identical argument can be sketched when $|\mathcal{S}_{\lambda}| < 2s$). Since  $(\mathbf{A},\mathbf{C})$ is $2s$-sparse observable, the pair $(\mathbf{A},\mathbf{C}_{\mathcal{V}\setminus\mathcal{S}_{\lambda}})$ is observable. However, based on Proposition  \ref{prop:geommult1}, this requires $\lambda$ to be observable w.r.t. at least one sensor in $\mathcal{V}\setminus\mathcal{S}_{\lambda}$, leading to the desired contradiction. This completes the proof of necessity. For sufficiency, note from Proposition  \ref{prop:implication1} that the pair $(\mathbf{A},\mathbf{C})$ is at least $2s$-sparse observable whenever its eigenvalue observability index is $2s$; the fact that the observability index is no more than $2s$ follows from the additional assumption on the geometric multiplicity of eigenvalues, and arguments similar to those used for establishing necessity. 
\end{proof}

It directly follows from the definition of eigenvalue observability that the eigenvalue observability index of a system can be computed in polynomial time. Hence, we have the following corollaries of Proposition~\ref{PropEquivalence}.

\begin{corollary}
\label{cor1}
When all the eigenvalues of the matrix $\mathbf{A}$ have geometric multiplicity 1, the sparse observability index of the system can be computed in polynomial time.
\end{corollary}

\begin{corollary}
\label{cor2}
For a 2s-sparse observable system~\eqref{eqn:modelsys}-\eqref{eqn:modelmeasure}, when all the eigenvalues of the matrix $\mathbf{A}$ have geometric multiplicity 1, the SSR problem can be solved in polynomial time.
\end{corollary}
\begin{proof}
It is shown in Proposition \ref{prop:connection} that under the unitary geometric multiplicity assumption, a $2s$-sparse observable system is also $2s$-eigenvalue observable. Thus, such a system satisfies the hypotheses in the statement of Theorem \ref{the:poly}, and we immediately obtain the existence of a polynomial-time solution for the SSR problem.
\end{proof}

\section{CONCLUSION}
\label{sec:conclusion}

In this paper, we showed that when the eigenvalues of the system matrix $\mathbf{A}$ have unitary geometric multiplicity, the SSR problem is tractable since both checking the sparse observability (see Corollary~\ref{cor1}) as well as solving the SSR problem (see Theorem~\ref{the:poly}) can be performed in polynomial time. When at least one of the eigenvalues has geometric multiplicity greater than one, we can still compute the eigenvalue observability index and, if it is at least $2s$, solve the SSR problem in polynomial time if at most $s$ sensors are attacked. However, in this case, eigenvalue observability is no longer necessary for the SSR problem to be solvable. Since even checking sparse observability is coNP-complete, we conjecture that the SSR problem may be intractable in this case. The authors are currently investigating this conjecture. However, even in this case, the computational complexity of solving the SSR problem can be reduced, when the system matrix $\mathbf{A}$ has at least 2 distinct eigenvalues.

\section{ACKNOWLEDGMENTS}
Shreyas Sundaram thanks Lintao Ye for helpful discussions pertaining to the Linear Degeneracy problem.

\bibliographystyle{unsrt}
\bibliography{refs}

\end{document}